\documentclass{amsart}
     
\newtheorem{theorem}{Theorem}[section]
\newtheorem{lemma}[theorem]{Lemma}

\theoremstyle{definition}

\theoremstyle{remark}
\newtheorem{remark}[theorem]{Remark}
 \usepackage{epsfig} 
 \usepackage{graphics} 
\numberwithin{equation}{section}

\def\a{{\alpha}}

\def\g{{\gamma}}

\def\e{{\epsilon}}
\def\ep{{\eta}}

\def\f{{\varphi}}

\def\r{{\rho}}

\def\p{{\psi}}

\def\G{{\Gamma}}
\def\O{{\Omega}}

\def\vvv{{{\bf v}}}

\def\R{{{\bf R}^1}}
\def\RR{{{\bf R}^2}}
\def\RRR{{{\bf R}^3}}

\def\9{{\ \hbox{in}\ \O}}
\def\1{{\ \hbox{on}\ \G_1}}
\def\2{{\ \hbox{on}\ \G_2}}
\def\3{{\ \hbox{on}\ \G_3}}
\def\0{{\ \hbox{on}\ \G}}

\begin{document}

\title[Boundary Value Problems in General Relativity]{Boundary Value Problems in General Relativity}
\author{Giovanni Cimatti}
\address{Department of Mathematics, Largo Bruno
  Pontecorvo 5, 56127 Pisa Italy}
\email{cimatti@dm.unipi.it}






\begin{abstract}

Certain theorems of existence, non-existence and uniqueness for boundary value problems modeling axial symmetric problems in general relativity are presented  using the Weyl's metric. A solution related to the classical Poiseuille solution of non-relativistic fluid mechanics is also presented.
\end{abstract}

\maketitle

\section{Introduction}
The theory of newtonian potential is essentially a theory of elliptic boundary value problems. On the other hand, in general relativity the coefficients of the metric play the role of unknown potentials \cite{WP}, but the corresponding boundary value problems seems to have received little attention. In this paper we consider, in Section 1, the simple case in which the energy-stress tensor vanishes everywhere and we give a theorem of existence and uniqueness for the corresponding boundary value problem stated in a bounded axial symmetric domain when the values of the coefficients of the metric are given on the boundary. In Section 3 we treat the case of the energy-stress tensor corresponding to a continuous distribution of fluid in absence of body forces and fluid motion. We obtain a result of existence and of non-existence of solutions. Finally in Section 4 we find an analogue of the Poiseuille solution of classical fluid mechanics starting from the Einstein's equations. We always consider axial symmetric situations and use the Weyl's metric \cite{HW}

\begin{equation}
\label{-1_1}
ds^2=e^{2\p}dt^2-e^{2\g-2\p}d\r^2-e^{2\g-2\p}dz^2-e^{-2\p}\r^2d\f^2,
\end{equation}
where the unknown ``potentials'' $\p$ and $\g$ are assumed to be functions of $\r$ and $z$ only $\p=\p(\r,z)$, $\g=\g(\r,z)$. We recall that the non-vanishing components of the Einstein's tensor $G_{ik}$ corresponding to the metric (\ref{-1_1}) are \cite{SKMHH}

\begin{equation}
\label{1_1}
G_{11}=e^{4\p+2\g}\bigl[-2\bigl(\p_{\r\r}+\frac{\p_\r}{\r}+\p_{zz}\bigl)+\p_\r^2+\p_z^2+\g_{\r\r}+\g_{zz}\bigl]
\end{equation}

\begin{equation}
\label{2_1}
G_{22}=\p_\r^2-\frac{\g_\r}{\r}-\p_z^2,\quad G_{23}=2\p_\r\p_z-\frac{\g_z}{\r},\quad G_{33}=-\p_\r^2+\frac{\g_\r}{\r}+\p_z^2
\end{equation}

\begin{equation}
\label{5_1}
G_{44}=-e^{-2\g}\r^2\bigl(\p_\r^2+\p_z^2+\g_{\r\r}+\g_{zz}\bigl).
\end{equation}

\section{The case $T_{ik}=0$}
Let us consider an axial symmetric bounded subset $\O$ of $\RRR$ with a regular boundary $\G$. We suppose $\O$ free of gravitational masses and of electric charges. Outside $\O$ there exists a distribution of mass, also  axial symmetric, which determines on $\G$ the values of $\p$ and $\g$ in a way not depending on angular variable $\f$. We assume the energy-stress tensor $T_{ik}$ to vanish in $\O$. The Einstein's equations

\begin{equation}
\label{1_2}
G_{ik}=K\ T_{ik},\quad K=8\pi\frac{G}{c^4},\quad \hbox{G gravitational constant}
\end{equation}
determine the ``potentials'' $\p(\r,z)$ and $\g(\r,z)$ via the following overdetermined system of partial differential equations  \cite{WB}, \cite{RW}

\begin{equation}
\label{1_3}
-2\bigl(\p_{\r\r}+\frac{\p_\r}{\r}+\p_{zz}\bigl)+\p_\r^2+\p_z^2+\g_{\r\r}+\g_{zz}=0
\end{equation}

\begin{equation}
\label{2_3}
\p_\r^2-\frac{\g_\r}{\r}-\p_z^2=0
\end{equation}

\begin{equation}
\label{4_3}
2\p_\r\p_z-\frac{\g_z}{\r}=0
\end{equation}

\begin{equation}
\label{5_3}
\p_\r^2+\p_z^2+\g_{\r\r}+\g_{zz}=0.
\end{equation}
The system (\ref{1_3})-(\ref{5_3}) is supplemented with the boundary conditions

\begin{equation}
\label{6_3}
\p(\r,z)=\p_\G,\quad\g(\r_0,z_0)=\g_0,\quad (\r_0,z_0)\in\G,\quad \g_0\in\R.
\end{equation}
We prove in the next theorem that the boundary value problem (\ref{1_3})-(\ref{6_3}) has one and only one solution. In (\ref{6_3}) $\p_\G$ is a function given on $\G$ whereas $\g$ is supposed to be known in one point of $\G$. We prove in fact the more general 

\begin{theorem}
Let the section of $\O$ with an arbitrary half-plane containing the $z$-axis be simply connected. If the function $\r h(\r,z)+ig(\r,z)$  of the complex variable $\r+iz$ \footnote{We are in the physically relevant case when $g=0$ and $h=0$.}is analytic then the overdetermined system of P.D.E.

\begin{equation}
\label{1_5}
-2\bigl(\p_{\r\r}+\frac{\p_\r}{\r}+\p_{zz}\bigl)+\p_\r^2+\p_z^2+\g_{\r\r}+\g_{zz}=0
\end{equation}

\begin{equation}
\label{2_5}
\p_\r^2-\frac{\g_\r}{\r}-\p_z^2=g(\r,z)
\end{equation}

\begin{equation}
\label{4_5}
2\p_\r\p_z-\frac{\g_z}{\r}=h(\r,z)
\end{equation}

\begin{equation}
\label{5_5}
\p_\r^2+\p_z^2+\g_{\r\r}+\g_{zz}=0,
\end{equation}
with the boundary conditions

\begin{equation}
\label{7_5}
\p=\p_\G \quad\0
\end{equation} 

\begin{equation}
\label{8_5}
\g(\r_0,z_0)=\g_0,\quad (\r_0,z_0)\in\RR,\quad \g_0\in\R,
\end{equation} 
has one and only one solution.
\end{theorem}
For the proof of Theorem 2.1 use will be made of the following

\begin{lemma}
Let $\O$ be an axial symmetric subset of $\RRR$ with a regular boundary $\G$. Let $\p_\G$ be a function of class $C^{0,\a}(\G)$ not depending on the angular variable. Then the problem

\begin{equation}
\label{1_6}
\p_{\r\r}+\frac{\p_\r}{\r}+\p_{zz}=0\ \9,\quad\p=\p_{\G}\0
\end{equation}
has one and only one solution.
\end{lemma}

\begin{proof}
The operator $\p_{\r\r}+\frac{\p_\r}{\r}+\p_{zz}$ entering in (\ref{1_6}) is a ``piece'' of the laplacian in cylindrical coordinates. This suggests to consider the auxiliary problem

\begin{equation}
\label{2_7}
\p_{\r\r}+\frac{\p_\r}{\r}+\p_{zz}+\frac{\p_{\f\f}}{\r^2}=0\ \9,\quad\p=\p_{\G}\0.
\end{equation}
By standard results, see e.g.\cite{PW}, the problem (\ref{2_7}) has one and only one solution $\p(\r,z,\f)$. On the other hand, $\bar\p(\r,z,\f+k)$, with $k$ an arbitrary real number, is also a solution of the problem (\ref{2_7}) in view of the axial symmetry of $\O$ and $\p_\G$. Hence, by uniqueness, $\bar\p=\p$. Therefore, $\p$ does not depend on $\f$ and thus it is the only solution of the original problem (\ref{1_6}).
\end{proof}

\begin{proof}(of Theorem 2.1) Let $\tilde\p(\r,z)$ be the unique solution of problem (\ref{1_6}) given by Lemma 2.2 and let us consider the first order system

\begin{equation}
\label{1_9}
\g_\r=F(\r,z),\quad\g_z=G(\r,z)
\end{equation}
with the condition

\begin{equation}
\label{3_9}
\g\bigl(\r_0,z_0\bigl)=\g_0,
\end{equation}
where in (\ref{1_9}) $F(\r,z)=\r\bigl(\tilde\p_\r^2-\tilde\p_z^2+g(\r,z)\bigl),\quad G(\r,z)=2\r\bigl(\tilde\p_\r\tilde\p_z+h(\r,z)\bigl)$. We have

\begin{equation*}
G_\r-F_z=2\r\tilde\p_z\bigl(\tilde\p_{\r\r}+\frac{\tilde\p_\r}{\r}+\tilde\p_{zz}\bigl)+\bigl(\r h)_\r-\bigl(\r g\bigl)_z.
\end{equation*}
Since $\r h+i\r g$ is analytic we conclude that $G_\r=F_z$ by (\ref{1_6}). Therefore, the system (\ref{1_9}) is integrable and with the condition (\ref{3_9}) its solution is unique. It remains to prove that

\begin{equation}
\label{1_10}
\tilde\p_\r^2+\tilde\p_z^2+\tilde\g_{\r\r}+\tilde\g_{zz}=0.
\end{equation}
From (\ref{1_9}) we have

\begin{equation*}
\tilde\g_{\r\r}+\tilde\g_{zz}=\tilde\p_\r^2-\tilde\p_z^2+2\r\tilde\p_\r\tilde\p_{\r\r}+2\r\tilde\p_\r\tilde\p_{zz}+\bigl(g+\r g_\r+\r h_z\bigl).
\end{equation*}
On the other hand, $\r h+i\r g$ is analytic, hence

\begin{equation}
\label{3_10}
\tilde\g_{\r\r}+\tilde\g_{zz}=\tilde\p_\r^2-\tilde\p_z^2+2\r\tilde\p_\r\tilde\p_{\r\r}+2\r\tilde\p_\r\tilde\p_{zz}.
\end{equation}
Substituting (\ref{3_10}) in (\ref{1_10}) we have, by (\ref{1_6}),

\begin{equation*}
\tilde\p_\r^2+\tilde\p_z^2+\tilde\g_{\r\r}+\tilde\g_{zz}=2\tilde\p_\r\r\bigl(\tilde\p_{\r\r}+\frac{\tilde\p_\r}{\r}+\tilde\p_{zz}\bigl)=0.
\end{equation*}
Therefore, $(\tilde\p(\r,z),\tilde\g(\r,z))$ is a solution of the problem (\ref{1_5})-(\ref{8_5}). On the other hand, this solution is also unique. For, let $(\p^*,\g^*)$ be a second solution. By difference from (\ref{1_3}) and (\ref{5_3}) we obtain

\begin{equation*}
\p^*_{\r\r}+\frac{\p^*_\r}{\r}+\p^*_{zz}=0\ \9,\quad\p^*=\p_{\G}\0.
\end{equation*}
Thus $\tilde\p$ and $\p^*$ are both solutions of a problem for which there is uniqueness. Hence $\tilde\p=\p^*$. Since $\tilde F(\r,z)=F^*(\r,z)$, $\tilde G(\r,z)=G^*(\r,z)$ and $\tilde\g(\r_0,z_0)=\g^*(\r_0,z_0)$, we also have $\tilde\g(\r,z)=\g^*(\r,z)$. We conclude that the problem (\ref{1_5})-(\ref{8_5}) has one and only one solution.
\end{proof}

\section{The boundary value problem in presence of a fluid}
In this Section we assume again to be in the axial symmetric situation of Section 1 with the metric (\ref{-1_1}).  The domain $\O$ is supposed to be filled with an incompressible viscous fluid. The energy-stress tensor reads (\cite{LL} page 512)

\begin{equation}
\label{1_14}
T_{ik}=-pg_{ik}+(p+\e)u_iu_k-c\ep\bigl(u_ {i;k}+u_{k;i}-u_iu^lu_{k;l}-u_ku^lu_{i;k}\bigl),
\end{equation}
where $u_i$ is the covariant four-velocity, $p$ denotes the pressure, $\ep$ is the viscosity and $\e$ the energy density which is assumed here to be a given constant.  We consider first an hydrostatic case in which the physical velocity $\vvv=(v_\r,v_z,v_\f)$ is assumed to vanish. Correspondingly the covariant four-velocity is given, recalling the metric (\ref{-1_1}), by

\begin{equation}
\label{1_15}
u_i=\bigl(e^\p,0,0,0\bigl).
\end{equation}
We want to determine $\p$, $\g$ and $p$ assuming the boundary conditions (\ref{6_3}). The non-vanishing components of the covariant derivative $u_{i;j}$ are

\begin{equation*}
u_{2;1}=-\p_\r e^\p,\quad u_{3;1}=-\p_z e^\p
\end{equation*}
and those of the energy-stress tensor $T_{ik}$ are given by

\begin{equation*}
T_{11}=-pe^{2\p}+(\e+p)e^{2\p},\quad T_{22}=pe^{2\g-2\p},\quad T_{33}=pe^{2\g-2\p},\quad T_{44}=pe^{-2\p}\r^2.
\end{equation*}
The Einstein's equations (\ref{1_2}) become,  by (\ref{1_1})-(\ref{5_1}),

\begin{equation}
\label{1_17}
e^{4\p-2\g}\bigl[\p_\r^2+\p_z^2+\g_{\r\r}+\g_ {zz}-2\bigl(\p_{\r\r}+\frac{\p_\r}{\r}+\p_{zz}\bigl)\bigl]=K\bigl[\bigl(\e+p\bigl)e^{2\p}-pe^{2\p}\bigl]
\end{equation}

\begin{equation}
\label{2_17}
\p_\r^2-\p_z^2-\frac{\g_\r}{\r}=Kpe^{2\g-2\p}
\end{equation}

\begin{equation}
\label{3_17}
\frac{\g_\r}{\r}-\p_\r^2+\p_z^2=Kpe^{2\g-2\p}
\end{equation}

\begin{equation}
\label{4_17}
2\p_r\p_z-\frac{\g_z}{\r}=0
\end{equation}

\begin{equation}
\label{5_17}
-e^{-2\g}\r^2\bigl(\p_\r^2+\p_z^2+\g_{\r\r}+\g_{zz}\bigl)=K p e^{-2\p}\r^2.
\end{equation}
By adding (\ref{2_17}) and (\ref{3_17}) we infer

\begin{equation*}
p=0.
\end{equation*}
Hence the system (\ref{1_17})-(\ref{5_17}) can be rewritten as

\begin{equation}
\label{2_18}
\p_\r^2+\p_z^2+\g_{\r\r}+\g_{zz}-2\bigl(\p_{\r\r}+\frac{\p_\r}{\r}+\p_{zz}\bigl)=K\e e^{2\g-2\p}
\end{equation}

\begin{equation}
\label{3_18}
\g_\r=\r\bigl(\p_\r^2-\p_z^2\bigl)
\end{equation}

\begin{equation}
\label{4_18}
\g_z=2\r\p_\r\p_z
\end{equation}

\begin{equation}
\label{5_18}
\p_\r^2+\p_z^2+\g_{\r\r}+\g_{zz}=0.
\end{equation}
From (\ref{2_18}) and (\ref{5_18}) we obtain

\begin{equation}
\label{6_18}
-2\bigl(\p_{\r\r}+\frac{\p_\r}{\r}+\p_{zz}\bigl)=K\e\  e^{2\g-2\p}.
\end{equation}
Let $F(\r,z)=\r\bigl(\p_\r^2-\p_z^2\bigl),\quad G(\r,z)=2\r\p_\r\p_z$. From (\ref{3_18}) and (\ref{4_18}) we have

\begin{equation*}
F_z-G_\r=-2\r\p_z\bigl(\p_{\r\r}+\frac{\p_\r}{\r}+\p_{zz}\bigl).
\end{equation*}
Hence, by (\ref{6_18}) we obtain

\begin{equation*}
F_z-G_\r=2K\e e^{2\g-2\p}\r\p_z.
\end{equation*}
Therefore, the system (\ref{3_18}), (\ref{4_18}) i.e.

\begin{equation*}
\g_\r=F(\r,z),\quad \g_z=G(\r,z)
\end{equation*}
is not integrable if $\e\neq 0$ and the system (\ref{1_17})-(\ref{5_17}) cannot have solutions in this case\footnote{This implies that the condition $\e\neq 0$ is incompatible with the assumption (\ref{1_15}) of absence of fluid motion}. On the other hand, if $\e=0$ the system (\ref{2_18})-(\ref{5_18}) becomes

\begin{equation*}
\p_\r^2+\p_z^2+\g_{\r\r}+\g_{zz}-2\bigl(\p_{\r\r}+\frac{\p_\r}{\r}+\p_{zz}\bigl)=0
\end{equation*}

\begin{equation*}
\g_\r=\r\bigl(\p_\r^2-\p_z^2\bigl)
\end{equation*}

\begin{equation*}
\g_z=2\r\p_\r\p_z
\end{equation*}

\begin{equation*}
\p_\r^2+\p_z^2+\g_{\r\r}+\g_{zz}=0.
\end{equation*}
Thus we are formally in the situation of Section 1 and the corresponding results of existence and uniqueness apply.

\section{A Poiseuille-like solution in general relativity}
The space outside an indefinite cylinder of radius $R$ is filled with an incompressible viscous fluid\footnote{The axis of the cylinder is the $z$-axis of the cylindrical coordinates system}. Body forces are absent and the cylinder is supposed to move with a constant velocity $v_R$ in the $z$-direction. If we assume

\begin{equation}
\label{1_7n}
v_\r=0,\quad v_z=v(\r),\quad v_\f=0
\end{equation}
the Navier-Stockes equations reduce to the single equation, see \cite{LL1}

\begin{equation}
\label{3_7n}
v''+\frac{v'}{\r}=0.
\end{equation}
The non-slip condition gives

\begin{equation}
\label{1_8n}
v(R)=v_R.
\end{equation}
 In view of (\ref{1_8n}) we have, solving (\ref{3_7n})

\begin{equation*}
v(\r)=C\log\r+v_R-C\log R.
\end{equation*}
Assuming the velocity to remain bounded we obtain the exceedingly simple solution

\begin{equation*}
v(\r)=v_R.
\end{equation*}
We ask the following question: starting with the same assumptions (\ref{1_7n}), (\ref{1_8n}) which solution is given by the equations of general relativity? We use the Weyl's metric (\ref{1_1}) assuming $\p$ and $\g$ to be function of $\r$ only. In view of the one-dimensional character of the problem, the non-vanishing components of the Einstein tensor take on the form

\begin{equation*}
G_{11}=e^{4\p+2\g}\Bigl[\p'^2+\g''-2\Bigl(\p''+\frac{\p'}{\r}\Bigl)\Bigl],\quad G_{22}=\p'^2-\frac{\g'}{\r}
\end{equation*}

\begin{equation*}
G_{33}=-\p'^2+\frac{\g'}{\r},\quad G_{44}=-e^{-2\g}\r^2\bigl(\p'^2+\g''\bigl).
\end{equation*}
The contravariant four-velocity $u^i=\frac{dx_i}{ds}=\bigl(e^{-\p},0,v(\r)e^{-\p},0\bigl)$ in covariant form is
\begin{equation*}
u_i=\bigl(e^\p,0,-e^{2\g-3\p},0\bigl).
\end{equation*}
 Turning to the non vanishing components of the energy-stress tensor we find, using (\ref{1_14})

\begin{equation*}
T_{11}=\e e^{2\p},\quad T_{12}=c\eta v^2 e^{2\g-3\p}\bigl(\g'-\p'\bigl),\quad T_{13}=-\bigl(\e+p\bigl)v e^{2\g-2\p}
\end{equation*}

\begin{equation*}
T_{22}=p e^{2\g-2\p}, \quad T_{23}=c\eta v' e^{2\g-3\p}
\end{equation*}

\begin{equation*}
T_{33}=pe^{2\g-2\p}+(\e+p)e^{4\g-6\p}v^2,\quad T_{44}=p\r e^{-2\p}.
\end{equation*}
For the non identically satisfied Einstein's equations (\ref{1_2}) we find

\begin{equation}
\label{5_12n}
e^{2\p-2\g}\bigl[\p'^2+\g''-2\bigl(\p''+\frac{\p'}{\r}\bigl)\bigl]=\e K
\end{equation}

\begin{equation}
\label{6_12n}
c\eta v^2e^{2\g-3\p}\bigl(\g'-\p'\bigl)=0
\end{equation}

\begin{equation}
\label{1_13n}
(\e+p)e^{2\g-2\p}v=0
\end{equation}

\begin{equation}
\label{2_13n}
\p'^2-\frac{\g'}{\r}=Kpe^{2\g-2\p}
\end{equation}

\begin{equation}
\label{3_13n}
c\eta v' e^{2\g-3\p}=0
\end{equation}

\begin{equation}
\label{4_13n}
-\p'^2+\frac{\g'}{\r}=K\bigl[pe^{2\g-2\p}+(\e+p)e^{4\g-6\p}v^2\bigl]
\end{equation}

\begin{equation}
\label{5_13n}
\p'^2+\g''=-Kp e^{2\g-2\p}.
\end{equation}
From (\ref{3_13n}) we have $v(\r)=v_R$, as in the Navier-Stockes case, if we assume the boundary condition $v(R)=v_R$. From (\ref{6_12n}) we obtain,  if  $v_R\neq 0$,

\begin{equation}
\label{1_14n}
 \g=\p+C
\end{equation}
 $C$ a constant to be determined. By (\ref{1_13n}) we have

\begin{equation*}
p=-\e.
\end{equation*}
Thus (\ref{4_13n}) becomes

\begin{equation}
\label{3_14n}
-\p'^2+\frac{\g'}{\r}=K p e^{2\g-2\p}.
\end{equation}
Adding (\ref{3_14n}) and (\ref{2_13n}) we obtain

\begin{equation}
\label{4_14n}
2K p e^{2\g-2\p}=0.
\end{equation}
Hence

\begin{equation}
\label{5_14n}
p=0,\quad \e=0.
\end{equation}
We conclude that the problem is compatible only with a null energy density. Moreover, (\ref{5_13n}) becomes

\begin{equation}
\label{1_15n}
\p'^2+\g''=0.
\end{equation}
In addition (\ref{5_12n}) gives

\begin{equation*}
\p''+\frac{\p'}{\r}=0.
\end{equation*}
This means

\begin{equation*}
\p(\r)=k_1\log\r+k_2.
\end{equation*}
On the other hand, by (\ref{1_14n}) we have

\begin{equation*}
\p'=\g'=\frac{k_1}{\r},\quad \g''=-\frac{k_1}{\r^2}.
\end{equation*}
By (\ref{1_15n}), it follows

\begin{equation*}
\frac{k_1^2}{\r^2}=\frac{k_1}{\r^2}.
\end{equation*}
Hence, either $k_1=1$ or $k_1=0$. If $k_1=1$ we have

\begin{equation*}
\p(\r)=\log\r+k_2.
\end{equation*}
The constant $k_2$ is determined with a condition of the form $\p(R)=\p_R$ which gives $\p(\r)=\log\r-\log R+\p_R$ (compare for this solution \cite{TLC}). Finally with a condition like $\g(R)=\g_R$ we determine the, still unknown, constant $C$ entering in (\ref{1_14n}). We obtain $\g(\r)=\log\r-\log R+\g_R$. We conclude that the boundary conditions $\p(R)=\p_R$, $ \g(R)=\g_R$ and $v(R)=v_R$ make the problem well-posed in agreement with the corresponding solution of the Navier-Stockes equations and  determine completely also the metric (\ref{-1_1}).

\begin{remark}
Observe that the equations of motion $T^{ij}_{;j}=0$ is in this case automatically satisfied. 
\end{remark}

\begin{remark}
If, in the contest of Newtonian hydrodynamics, we state the cognate problem of the viscous fluid motion between two coaxial cylinders of radii $R_2>R_1>0$ moving respectively with given velocities $v_2$ and $v_1$ in the $z$ direction, again in absence of body forces, the problem is well-posed and easily solved with the non-slip conditions $v(R_2)=v_2$, $v(R_1)=v_1$ and we find easily

\begin{equation*}
v(\r)=\frac{v_2-v_1}{\log\frac{\r}{R_1}}\log\frac{\r}{R_1}+v_1.
\end{equation*}
However, in the contest of general relativity there is no parallel situation since the equation (\ref{3_13n}) permits to impose only one boundary condition. This  is inherent to the fact that the Einstein's equations are of the first order.
\end{remark}

\bibliographystyle{amsplain}

\end{document}